%% file: main.tex
\documentclass[a4paper,USenglish]{lipics}

\usepackage{microtype}

\usepackage[inline]{enumitem}

\usepackage{tikz}
\usetikzlibrary{calc}
\usetikzlibrary{positioning}
\usetikzlibrary{shapes}
\usetikzlibrary{fit}
\tikzset{anonymous/.style={draw, circle, inner sep=0.7mm}}
\tikzset{forbidden/.style={draw, rectangle, inner sep=1mm}}
\tikzset{necessary/.style={draw, isosceles triangle, isosceles triangle apex angle=60, shape border rotate=90, minimum width=2.5mm, inner sep=0.5mm}}

\usepackage{circle} 
\newcommand{\medcirc}{{\Circle}}

\newcommand{\decisionProblem}[3]{%
\begin{center}
\fbox{%
\begin{minipage}{.95\linewidth}%
#1
\begin{itemize}[topsep=1mm,itemsep=1mm]
\addtolength{\leftskip}{11mm}
\item[Input:] #2
\item[Question:] #3
\end{itemize}%
\end{minipage}%
}%
\end{center}
}

\renewcommand{\leq}{\leqslant}
\renewcommand{\geq}{\geqslant}
\renewcommand{\phi}{\varphi}
\renewcommand{\theta}{\vartheta}
\newcommand{\size}[1]{\ensuremath{{\lvert #1 \rvert}}}

\newcommand{\calT}{\mathcal{T}}

\newcommand{\vertices}[1]{\ensuremath{\operatorname{V}(#1)}}
\newcommand{\edges}[1]{\ensuremath{\operatorname{E}(#1)}}

\newcommand{\prob}[1]{\textsc{#1}}

\newcommand{\probSS}{\prob{Secure Set}}

\newcommand{\SSFN}{\prob{Secure Set\textsuperscript{FN}}}
\newcommand{\SSFNC}{\prob{Secure Set\textsuperscript{FNC}}}

\newcommand{\DA}{\prob{Defensive Alliance}}
\newcommand{\DAF}{\prob{Defensive Alliance\textsuperscript{F}}}
\newcommand{\DAFN}{\prob{Defensive Alliance\textsuperscript{FN}}}
\newcommand{\DAFNC}{\prob{Defensive Alliance\textsuperscript{FNC}}}
\newcommand{\EDA}{\prob{Exact Defensive Alliance}}
\newcommand{\EDAF}{\prob{Exact Defensive Alliance\textsuperscript{F}}}
\newcommand{\EDAFN}{\prob{Exact Defensive Alliance\textsuperscript{FN}}}
\newcommand{\EDAFNC}{\prob{Exact Defensive Alliance\textsuperscript{FNC}}}

\newcommand{\MMO}{\prob{Minimum Maximum Outdegree}}

\renewcommand{\O}{\ensuremath{\mathcal{O}}}

\newcommand{\complexityclass}[1]{\ensuremath{\mathsf{#1}}}

\newcommand{\NP}{\complexityclass{NP}}

\newcommand{\FPT}{\complexityclass{FPT}}
\newcommand{\Wone}{\complexityclass{W[1]}}
\newcommand{\XP}{\complexityclass{XP}}

\newcommand{\sFNC}[1]{\ensuremath{\sigma_{#1}^\text{FNC}}}
\newcommand{\tFNC}{\ensuremath{\tau^\text{FNC}}}
\newcommand{\sFN}[1]{\ensuremath{\sigma_{#1}^\text{FN}}}
\newcommand{\tFN}{\ensuremath{\tau^\text{FN}}}
\newcommand{\tF}{\ensuremath{\tau^\text{F}}}

\title{Defensive Alliances in Graphs of Bounded Treewidth\footnote{This work was supported by the Austrian Science Fund (FWF) projects P25607 and Y698.}}

\makeatletter
\@ifclassloaded{lipics}{
\theoremstyle{plain}
\newtheorem{observation}[theorem]{Observation}

\author{Bernhard Bliem}
\author{Stefan Woltran}
\affil{%
Institute of Information Systems 184/2\\
TU Wien\\
Favoritenstrasse 9--11, 1040 Vienna, Austria\\
\texttt{[bliem,woltran]@dbai.tuwien.ac.at}%
}
\authorrunning{B. Bliem and S. Woltran} 

\Copyright{Bernhard Bliem and Stefan Woltran}

\subjclass{F.2.2 Nonnumerical Algorithms and Problems}
\keywords{defensive alliance, alliances in graphs, treewidth, complexity analysis, parameterized complexity}

}{
\usepackage[a4paper,margin=3.5cm]{geometry}
\usepackage{amssymb}
\usepackage[tbtags,fleqn]{amsmath}
\usepackage{amsthm}
\theoremstyle{plain}
\newtheorem{theorem}{Theorem}
\newtheorem{lemma}[theorem]{Lemma}
\newtheorem{corollary}[theorem]{Corollary}
\newtheorem{observation}[theorem]{Observation}
\theoremstyle{definition}
\newtheorem{definition}[theorem]{Definition}

\theoremstyle{remark}

\ifx\numberwithinsect\relax
  \@addtoreset{theorem}{section}
  \edef\thetheorem{\expandafter\noexpand\thesection\@thmcountersep\@thmcounter{theorem}}
\fi

\ifnum\pdfshellescape=1
  \usetikzlibrary{external}
  \immediate\write18{mkdir -p out/out}
  \immediate\write18{touch out/out/directory-exists}
  \IfFileExists{out/out/directory-exists}{
    \tikzsetexternalprefix{out/}
  }{}
  \tikzexternalize
\fi

\author{Bernhard Bliem and Stefan Woltran}
}
\makeatother

\begin{document}

\maketitle

\begin{abstract}
A set $S$ of vertices of a graph is a defensive alliance if, for each element 
of $S$, the majority of its neighbors is in $S$.
The problem of finding a defensive alliance of minimum size in a given graph is 
\NP-hard and there are polynomial-time algorithms if certain parameters are 
bounded by a fixed constant.
In particular, fixed-parameter tractability results have been obtained for some 
structural parameters such as the vertex cover number.
However, for the parameter treewidth, the question of whether the problem is 
FPT has remained open.
This is unfortunate because treewidth is perhaps the most prominent graph 
parameter and has proven successful for many problems.
In this work, we give a negative answer by showing that the problem is 
\Wone-hard when parameterized by treewidth, which rules out FPT algorithms 
under common assumptions.
This is surprising since the problem is known to be FPT when parameterized by 
solution size and ``subset problems'' that satisfy this property usually tend 
to be FPT for bounded treewidth as well.
We prove \Wone-hardness by using techniques from a recent hardness result for 
the problem of finding so-called secure sets in a graph.
%
\end{abstract}

\section{Introduction}

The objective of many problems that can be modeled as graphs is finding a group
of vertices that together satisfy some property.
In this respect, one of the concepts that has been quite extensively studied is 
the notion of a defensive alliance~
\cite{kristiansen2004alliances,kristiansen2002introduction}, which is a set of 
vertices such that for each element $v$ at least half of its neighbors are also 
in the alliance.
The name ``defensive alliance'' stems from the intuition that the neighbors of 
an element $v$ that are also in the alliance can help out in case $v$ is 
attacked by its other neighbors.

Notions like this can be applied to finding groups of nations, companies or 
individuals that depend on each other, but also to more abstract situations 
like finding groups of websites that form 
communities~\cite{DBLP:journals/computer/FlakeLGC02}.
Another possible application for defensive alliances are computer networks, 
where a defensive alliance represents computers that can provide a certain 
desired resource; any computer in an alliance can then, with the help of its 
neighbors that are also in the alliance, allow access to this resource from all 
of its neighbors simultaneously~\cite{DBLP:journals/combinatorics/HaynesHH03}.

Several variants of defensive alliances have also been studied.
The papers that originally proposed defensive alliances also propose related 
notions like offensive and powerful alliances.
An offensive alliance is a set $S$ of vertices such that every \emph{neighbor} 
of an element of $S$ has at least half of its neighbors in $S$, and a powerful 
alliance is both a defensive and an offensive alliance.
Any of these alliances is called \emph{global} if it is at the same time a 
dominating set.
Another variant is to consider alliances $S$ where, for each vertex $v \in S$, 
the difference between the number of neighbors of $v$ in $S$ and the number of 
other neighbors of $v$ is at most a given integer~\cite{ShafiqueD2003}.
For comprehensive overviews of different kinds of alliances in graphs, we refer 
to the surveys~\cite{yero2013defensive, fernau2014survey}.

The \DA{} problem can be specified as follows:
Given a graph $G$ and an integer $k$, is there a defensive alliance $S$
in $G$ such that $1\leq \size{S} \leq k$?
It is known that this problem is 
\NP-complete~\cite{Jamieson2007,jamieson2009algorithmic}, and so is the 
corresponding problem for global defensive alliances~\cite{CamiBDD06}.
However, if we restrict ourselves to trees, \DA{} becomes trivial and in fact 
the corresponding problems for several non-trivial variants become solvable in 
linear time~\cite{Jamieson2007}.

There has also been some work on the parameterized complexity of alliance 
problems.
In particular, determining whether a defensive, offensive and powerful alliance 
of a given (maximum) size exists is fixed-parameter tractable when 
parameterized by the solution size~\cite{sofsem:FernauR07, Enciso2009}.
Also structural parameters have been considered to some extent.
Recently, \cite{dam:KiyomiO17} proved that these problems can be solved in 
polynomial time if the clique-width of the instances is bounded by a constant.
The authors also provide an FPT algorithm when the parameter is the size of the 
smallest vertex cover.
Moreover, \cite{Enciso2009} showed that the decision problems for defensive 
alliances and global defensive alliances are fixed-parameter tractable when 
parameterized by the combination of treewidth and maximum degree.
Despite these advances regarding, the question of whether or not \DA{} 
parameterized by treewidth is fixed-parameter tractable has so far remained 
open.

Treewidth~\cite{DBLP:journals/jct/RobertsonS84,DBLP:journals/actaC/Bodlaender93,DBLP:conf/sofsem/Bodlaender05}
is one of the most extensively studied structural parameters and indicates how 
close a graph is to being a tree.
It is particularly attractive because many hard problems become tractable on 
instances of bounded treewidth, and in several practical applications it has 
been observed that the considered problem instances exhibit small 
treewidth~\cite{DBLP:journals/actaC/Bodlaender93,iandc:Thorup98,dam:KornaiT92}.
Hence it would be very appealing to obtain an FPT algorithm for the \DA{} 
problem using this parameter.

The main contribution of this paper is a parameterized complexity analysis of 
\DA{} with treewidth as the parameter.
The question of whether or not this problem is fixed-parameter tractable when 
parameterized by treewidth has so far been unresolved~\cite{dam:KiyomiO17}.
In the current chapter, we provide a negative answer to this question:
We show that the problem is hard for the class \Wone{}, which rules out 
fixed-parameter tractable algorithms under commonly held complexity-theoretic 
assumptions.
This result is rather surprising for two reasons:
First, the problem is tractable on trees~\cite{Jamieson2007} and quite often 
problems that become easy on trees turn out to become easy on graphs of bounded 
treewidth.%
%
\footnote{To be precise, \cite{Jamieson2007,ho2009rooted,ChangCHKLW12} show 
that some variants of \DA{} are tractable on trees, since \DA{} on trees is 
trivial.}
Second, this puts \DA{} among the very few ``subset problems'' that are 
fixed-parameter tractable w.r.t.\ solution size but not w.r.t.\ treewidth.
Problems with this kind of behavior are rather rare, as observed
by Dom et~al.~\cite{DBLP:conf/iwpec/DomLSV08}.

We show \Wone-hardness of the problem by first reducing a problem known to be 
\Wone-hard to a variant of \DA{}, where vertices can be forced to be in or out 
of every solution, and pairs of vertices can be specified to indicate that 
every solution must contain exactly one element of each such pair.
In order to prove the desired complexity result, we then successively reduce 
this variant to the standard \DA{} problem.

At the same time, we show \Wone{}-hardness for the exact variants of these 
problems, where we are interested in defensive alliances \emph{exactly} of a 
certain size.
Note that a set may lose the property of being a defensive alliance by adding 
or removing elements, so these are non-trivial problem variants.
Indeed, exact versions of alliance problems have also been mentioned as 
interesting variants in~\cite{sofsem:FernauR07} because some algorithms that 
work for the non-exact case stop to work for the exact case:
A graph has a defensive alliance of size at most $k$ if and only if it has a 
\emph{connected} defensive alliance of size at most $k$ since every component 
of a defensive alliance is itself a defensive alliance.
Algorithms that exploit this by looking only for connected solutions hence fail 
for the exact versions.
(In fact, we will also study the complexity of other problem variants where 
this connectedness property does not apply even in the non-exact case.)

This paper is organized as follows:
We first introduce our problems of interest and describe preliminary concepts 
in Section~\ref{sec:background}.
In Section~\ref{sec:da-wone-treewidth} we then show that the \DA{} problem is 
\Wone-hard when parameterized by treewidth.
Section~\ref{sec:conclusion} concludes the paper with a discussion.

The reductions in the current work are based on ideas used in the 
paper~\cite{DBLP:conf/wg/BliemW15}, which analyzed the complexity of a problem 
related to \DA{} called \probSS{}.
That paper has since been extended by a \Wone-hardness proof for the \probSS{} 
problem parameterized by treewidth~\cite{corr:BliemW14}.
In the current paper, we take up the ideas behind this hardness proof and apply 
them to the \DA{} problem.
Due to the different nature of these two problems, the reductions and proofs 
for \probSS{} do not work directly for \DA{} but require substantial 
modifications.

\section{Background}
\label{sec:background}

All graphs are undirected and simple unless stated otherwise.
We denote the set of vertices and edges of a graph $G$ by $\vertices{G}$ and $\edges{G}$,
respectively.
We denote an undirected edge between vertices $u$ and $v$ as $(u,v)$ or
equivalently $(v,u)$.
It will be clear from the context whether an edge $(u,v)$ is directed or
undirected.
Given a graph $G$, the \emph{open neighborhood} of a vertex $v \in \vertices{G}$,
denoted by $N_G(v)$, is the set of all vertices adjacent to $v$, and
$N_G[v] = N_G(v) \cup \{v\}$ is called the \emph{closed neighborhood} of $v$.
If it is clear from the context which graph is meant, we write $N(\cdot)$ and
$N[\cdot]$ instead of $N_G(\cdot)$ and $N_G[\cdot]$, respectively.

The intuition behind defensive alliances is the following:
If we consider a set $S$ of vertices as ``good'' vertices and all other 
vertices as ``bad'' ones, then $S$ being a defensive alliance means that each element of $S$ 
has at least as many ``good'' neighbors as ``bad'' neighbors.

\begin{definition}
Given a graph $G$, a set $S \subseteq \vertices{G}$ is a \emph{defensive 
alliance in $G$} if for each $v \in S$ it holds that
$\size{N[v] \cap S} \geq \size{N[v] \setminus S}$.
\end{definition}
We often write ``$S$ is a defensive alliance'' instead of ``$S$ is a defensive
alliance in $G$'' if it is clear from the context which graph is meant.
By definition, the empty set is a defensive alliance in any graph. Thus, in the
following decision problems we ask for a defensive alliances of size at
least~1.
\input{figures/da-example-figure}
For example, in Figure~\ref{fig:da-example}, the set $S = \{a,b\}$ is 
a defensive alliance as $\size{N[v] \cap S} \geq \size{N[v] \setminus S}$ 
holds for each $v \in S$.
Note that, for instance, $\{a,d\}$ is no defensive alliance since $d$ is 
attacked by three vertices but only has the neighbor $a$ to help defend itself.

Next we introduce several variants of \DA{} that we require in our proofs.
The problem \DAF{} generalizes \DA{} by designating some ``forbidden'' 
vertices that may never be in any solution.
This variant can be formalized as follows:
\decisionProblem{\DAF}{A graph $G$, an integer $k$ and a set $V_\square \subseteq \vertices{G}$}{%
Does there exist a set $S \subseteq \vertices{G} \setminus V_\square$ with $1 \leq \size{S} \leq k$ that is a defensive alliance?%
}
\DAFN{} is a further generalization that, in addition, allows ``necessary'' vertices to be specified that must occur in every solution.
\decisionProblem{\DAFN}{A graph $G$, an integer $k$, a set $V_\square \subseteq \vertices{G}$ and a set $V_\triangle \subseteq \vertices{G}$}{%
Does there exist a set $S \subseteq \vertices{G} \setminus V_\square$ with $V_\triangle \subseteq S$ and $1 \leq \size{S} \leq k$ that is a defensive alliance?%
}
Finally, we introduce the generalization \DAFNC{}.
Here we may state pairs of ``complementary'' vertices where each solution must contain exactly one element of every such pair.
\decisionProblem{\DAFNC}{A graph $G$, an integer $k$, a set $V_\square \subseteq \vertices{G}$, a set $V_\triangle \subseteq \vertices{G}$ and a set $C \subseteq \vertices{G}^2$}{%
Does there exist a set $S \subseteq \vertices{G} \setminus V_\square$ with $V_\triangle \subseteq S$ and $1 \leq \size{S} \leq k$ that is a defensive alliance and, for each pair $(a,b) \in C$, contains either $a$ or $b$?%
}
For our results on the parameter treewidth, we need a way to represent the 
structure of a \DAFNC{} instance by a graph that augments $G$ with the 
information in $C$:
\begin{definition}
Let $I$ be a \DAFNC\ instance, let $G$ be the graph in $I$ and let $C$ the set 
of complementary vertex pairs in $I$.
By the \emph{primal graph} of $I$ we mean the undirected graph $G'$ with
$\vertices{G'} = \vertices{G}$ and
$\edges{G'} = \edges{G} \cup C$.
\end{definition}
When we speak of the treewidth of an instance of \DA{}, \DAF{} or \DAFN{}, 
we mean the treewidth of the graph in the instance.
For an instance of \DAFNC{}, we mean the treewidth of the primal graph.

While the \DA{} problem asks for defensive alliances of size \emph{at most} $k$, we also consider the \EDA{} problem that concerns defensive alliances of size \emph{exactly} $k$.
Analogously, we also define exact versions of the three generalizations of 
\DA{} presented above.

In this paper's figures, we often indicate necessary vertices by means of a triangular node shape, and forbidden vertices by means of either a square node shape or a superscript square in the node name.
If two vertices are complementary, we often express this in the figures by putting a $\neq$ sign between them.
\input{figures/da-example-figure2}
For example, in Figure~\ref{fig:da-example2}, the vertices $b$ and $c$ are complementary and occur in no solution together; $a$ and the ``anonymous'' vertex adjacent to $c$ are necessary and occur in every solution; $d^\square$ and the ``anonymous'' vertex adjacent to $e$ are forbidden and occur in no solution.
In this figure, the unique minimum non-empty defensive alliance satisfying the conditions of forbidden, necessary and complementary vertices consists of $a$, $b$ and the ``anonymous'' necessary vertex adjacent to $c$.

The following terminology will be helpful:
We often use the terms \emph{attackers} and \emph{defenders} of an element $v$ of a defensive alliance candidate $S$.
By these we mean the sets $N[v] \setminus S$ and $N[v] \cap S$, respectively.
To show that an element $v$ of a defensive alliance candidate $S$ is \emph{not} a counterexample
to $S$ being a solution, we sometimes employ the notion of a \emph{defense} of
$v$ w.r.t.\ $S$, which assigns to each attacker a dedicated defender:
If we are able to find an injective mapping $\mu: N[v] \setminus S \to N[v] \cap S$, then obviously $\size{N[v] \setminus S} \leq \size{N[v] \cap S}$, and we call $\mu$ a \emph{defense} of $v$ w.r.t.\ $S$.
Given such a defense $\mu$, we say that a defender $d$ \emph{repels} an attack on $v$ by an attacker $a$ whenever $\mu(a) = d$.
Consequentially, when we say that a set of defenders $D$ \emph{can repel} attacks on $v$ from a set of attackers $A$, we mean that there is a defense that assigns to each element of $A$ a dedicated defender in $D$.

To warm up, we make some easy observations that we will use in our proofs.
First, for every set $R$ consisting of a majority of neighbors of a vertex $v$,
whenever $v$ is in a defensive alliance, also some element of $R$ must be in 
it:
\begin{observation}
\label{obs:one-in-s}
Let $S$ be a defensive alliance in a graph,
let $v \in S$ and
let $R \subseteq N(v)$.
If $\size{R} > \frac{1}{2} N[v]$,
then $S$ contains an element of $R$.
\end{observation}
\begin{proof}
Suppose that $\size{R} > \frac{1}{2} \size{N[v]}$ and
$S$ contains no element of $R$.
Since all elements of $R$ attack $v$,
$\size{N[v] \setminus S} > \frac{1}{2} \size{N[v]}$.
Hence
$2 \size{N[v] \setminus S} > \size{N[v]} =
\size{N[v] \cap S} + \size{N[v] \setminus S}$, and we obtain
the contradiction
$\size{N[v] \setminus S} > \size{N[v] \cap S}$.
\end{proof}

Next, if one half of the neighbors of an element $v$ of a defensive alliance 
attacks $v$, then the other half of the neighbors must be in the defensive 
alliance:
\begin{observation}
\label{obs:all-in-s}
Let $S$ be a defensive alliance in a graph,
let $v \in S$ and
let $N(v)$ be partitioned into two equal-sized sets $A,D$.
If $A \cap S = \emptyset$, then $D \subseteq S$.
\end{observation}
\begin{proof}
Since $N(v)$ is partitioned into $A$ and $D$ such that
$A \cap S = \emptyset$,
we get
$N(v) \cap S = D \cap S$.
If some element of $D$ is not in $S$, then
$D \cap S \subset D$ and
$A \subset N[v] \setminus S$.
By $\size{D} = \size{A}$, we get
$\size{D \cap S} + 2 \leq \size{N[v] \setminus S}$.
From
$\size{N[v] \cap S} = 1 + \size{N(v) \cap S} = 1 + \size{D \cap S}$
we now obtain the contradiction
$\size{N[v] \cap S} < \size{N[v] \setminus S}$.
\end{proof}
In particular, if half of the neighbors of $v$ are forbidden, then $v$ can only
be in a defensive alliance if all non-forbidden neighbors are also in the 
defensive alliance.

Finally, we recapitulate some background from complexity theory.
In parameterized complexity 
theory~\cite{downey1999parameterized,flum2006parameterized,b:Niedermeier06,Cygan15}, 
we study problems that consist not only of an input and a question, but also of 
some parameter of the input that is represented as an integer.
A problem is in the class \FPT{} (``fixed-parameter tractable'') if it can be 
solved in time $f(k) \cdot n^c$, where $n$ is the input size, $k$ is the 
parameter, $f$ is a computable function that only depends on $k$, and $c$ is a 
constant that does not depend on $k$ or $n$.
We call such an algorithm an \emph{FPT algorithm}, and we call it
\emph{fixed-parameter linear} if $c=1$.
Similarly, a problem is in the class \XP{} (``slice-wise polynomial'') if it 
can be solved in time $f(k) \cdot n^{g(k)}$, where $f$ and $g$ are computable 
functions.
Note that here the degree of the polynomial may depend on $k$, so such 
algorithms are generally slower than FPT algorithms.
For the class \Wone{} it holds that $\FPT \subseteq \Wone \subseteq \XP$,
and it is commonly believed that the inclusions are proper, i.e.,
\Wone{}-hard problems do not admit FPT algorithms.
\Wone{}-hardness of a problem can be shown using FPT reductions, 
which are reductions that run in FPT time and produce an equivalent instance
whose parameter is bounded by a function of the original parameter.

For problems whose input can be represented as a graph, one important parameter 
is \emph{treewidth}, which is a structural parameter that, roughly speaking, 
measures the ``tree-likeness'' of a graph. It is defined by means of tree 
decompositions, originally introduced in~\cite{DBLP:journals/jct/RobertsonS84}.
The intuition behind tree decompositions is to obtain a tree from a 
(potentially cyclic) graph by subsuming multiple vertices under one node and 
thereby isolating the parts responsible for cyclicity.
\begin{definition}
\label{def:td}
A \emph{tree decomposition} of a graph $G$
is a pair $\calT = (T,\chi)$ where $T$ is a (rooted) tree and
$\chi : \vertices{T} \to 2^{\vertices{G}}$ assigns to each node of $T$ a set of vertices of $G$
(called the node's \emph{bag}), such that the following conditions are met:
\begin{enumerate}
\item For every vertex $v \in \vertices{G}$, there is a node $t \in \vertices{T}$ such that $v 
\in \chi(t)$.
\item For every edge $(u,v) \in \edges{G}$, there is a node $t \in \vertices{T}$ such that
$\{u,v\} \subseteq \chi(t)$.
\item For every $v \in \vertices{G}$, the subtree of $T$ induced by $\{t \in \vertices{T} \mid 
v \in \chi(t)\}$ is connected.
\end{enumerate}
We call $\max_{t \in \vertices{T}} \lvert \chi(t) \rvert - 1$ the \emph{width} of $\calT$.
The \emph{treewidth} of a graph is the minimum width over all its tree 
decompositions.
\end{definition}
In general, constructing an optimal tree decomposition (i.e., a tree
decomposition with minimum width) is
intractable~\cite{Arnborg:1987:CFE:37170.37183}.
However, the problem is solvable in linear time on graphs of bounded treewidth
(specifically in time
$w^{\O(w^3)} \cdot n$, where $w$ is the
treewidth)~\cite{siamcomp:Bodlaender96} and
there are also heuristics that offer good performance in
practice~\cite{DBLP:conf/micai/DermakuGGMMS08,DBLP:journals/iandc/BodlaenderK10}.

In this paper we will consider so-called \emph{nice} tree decompositions:
\begin{definition}
\label{def:nice-td}
A tree decomposition $\mathcal{T} = (T,\chi)$ is \emph{nice} if each node $t 
\in \vertices{T}$ is of one of the following types:
\begin{enumerate}
 \item Leaf node: The node $t$ has no child nodes.
 \item Introduce node: The node $t$ has exactly one child node $t'$ such that 
 $\chi(t) \setminus \chi(t')$ consists of exactly one element.
 \item Forget node: The node $t$ has exactly one child node $t'$ such that 
 $\chi(t') \setminus \chi(t)$ consists of exactly one element.
 \item Join node: The node $t$ has exactly two child nodes $t_1$ and $t_2$ with 
 $\chi(t) = \chi(t_1) = \chi(t_2)$.
\end{enumerate}
Additionally, the bags of the root and the leaves of $T$ are empty.
\end{definition}
A tree decomposition of width $w$ for a graph with $n$ vertices can be 
transformed into a nice one of width $w$ with $\O(wn)$ nodes in fixed-parameter 
linear time~\cite{kloks1994treewidth}.

\input{figures/td-example-figure}

For any tree decomposition $\calT$ and an element $v$ of some bag in $\calT$, 
we use the notation $t_v^\calT$ to denote the unique ``topmost node'' whose bag 
contains $v$ (i.e., $t_v^\calT$ does not have a parent whose bag contains $v$).
Figure~\ref{fig:td-example} depicts a graph and a nice tree decomposition, 
where we also illustrate the $t^\calT_v$ notation.

When we speak of the treewidth of an instance of \DA{}, \DAF{}, \DAFN{}, \EDA{}, \EDAF{} or
\EDAFN{}, we mean the treewidth of the graph in the instance.
For an instance of \DAFNC{} or \EDAFNC{}, we mean the treewidth of the primal graph.

\section{Hardness of Defensive Alliance Parameterized by Treewidth}
\label{sec:da-wone-treewidth}

In this section, we prove the following theorem:
\begin{theorem}
\label{thm:da-variants-wone}
The following problems are all $\Wone$-hard when parameterized by treewidth:
\DA{},
\EDA{},
\DAF{},
\EDAF{},
\DAFN{},
\EDAFN{},
\DAFNC{}, and
\EDAFNC{}.
\end{theorem}

We prove hardness by providing a chain of FPT reductions from
a \Wone-hard problem to the problems under consideration.
Under the widely held assumption that $\FPT \neq \Wone$, this rules out
fixed-parameter tractable algorithms for these problems.

\subsection{Hardness of Defensive Alliance with Forbidden, Necessary and Complementary Vertices}

To show \Wone-hardness of \DAFNC{}, we reduce from the following 
problem~\cite{DBLP:journals/dam/AsahiroMO11}, which is known to be 
\Wone-hard~\cite{DBLP:journals/corr/abs-1107-1177} parameterized by the 
treewidth of the graph:
\decisionProblem{\MMO}{%
A graph $G$,
an edge weighting $w: \edges{G} \to \mathbb{N}^+$ given in unary
and a positive integer $r$}{%
Is there an orientation of the edges of $G$ such that, for each $v \in 
\vertices{G}$, the sum of the weights of outgoing edges from $v$ is at most 
$r$?%
}

\begin{lemma}
\DAFNC{} and \EDAFNC{}, both parameterized by the treewidth of the primal 
graph, are \Wone-hard.
\end{lemma}

\begin{proof}
Let an instance of \MMO\ be given by a graph $G$, an edge weighting $w: 
\edges{G} \to \mathbb{N}^+$ in unary and a positive integer $r$.
From this we construct an instance of both \DAFNC{} and \EDAFNC{}.
An example is given in Figure~\ref{fig:minmaxoutdegree-reduction-example}.
For each $v \in \vertices{G}$, we define the set of new vertices
$H_v = \{h^v_1, \dots, h^v_{2r - 1}\}$, and for each $(u,v) \in \edges{G}$, we 
define the sets of new vertices
$V_{uv} = \{u^v_1, \dots, u^v_{w(u,v)}\}$,
$V_{uv}^\square = \{u^{v\square}_1, \dots, u^{v\square}_{w(u,v)}\}$,
$V_{vu} = \{v^u_1, \dots, v^u_{w(u,v)}\}$ and
$V_{vu}^\square = \{v^{u\square}_1, \dots, v^{u\square}_{w(u,v)}\}$.
We now define the graph $G'$ with
\begin{align*}
  \vertices{G'} ={}& \vertices{G} \cup \bigcup_{v \in \vertices{G}} H_v
  \cup \bigcup_{(u,v) \in \edges{G}}
  (V_{uv} \cup V_{uv}^\square \cup V_{vu} \cup V_{vu}^\square),\\
  \edges{G'} ={}& \{(v,h) \mid v \in \vertices{G},\; h \in H_v\}\\
     {}\cup{}&\{(u,x) \mid (u,v) \in \edges{G},\;
       x \in V_{uv} \cup V_{uv}^\square\}\\
     {}\cup{}& \{(x,v) \mid (u,v) \in \edges{G},\;
       x \in V_{vu} \cup V_{vu}^\square\}.
\end{align*}
We also define the set of complementary vertex pairs
$C = \{(u^v_i,v^u_i) \mid (u,v) \in \edges{G},\; 1 \leq i \leq w(u,v)\} \cup
\{(v^u_i,u^v_{i+1}) \mid (u,v) \in \edges{G},\; 1 \leq i < w(u,v)\}$.
Finally, we define the set of necessary vertices
$V_\triangle = \vertices{G} \cup \bigcup_{v \in \vertices{G}} H_v$,
the set of forbidden vertices
$V_\square = \bigcup_{(u,v) \in \edges{G}} (V_{uv}^\square \cup 
V_{vu}^\square)$
and
$k = \size{V_\triangle} + \sum_{(u,v) \in \edges{G}} w(u,v)$.
We use $I$ to denote $(G',k,C,V_\triangle,V_\square)$, which is an instance of 
\DAFNC\ and also of \EDAFNC.

\input{figures/minmaxoutdegree-reduction-example}

Clearly $I$ can be computed in polynomial time.
We now show that the treewidth of the primal graph of $I$ depends only on the treewidth of $G$.
We do so by modifying an optimal tree decomposition $\calT$ of $G$ as follows:
\begin{enumerate}
\item For each $(u,v) \in \edges{G}$, we take an arbitrary node whose bag $B$ contains
both $u$ and $v$
and add to its children
a chain of nodes $N_1, \dots, N_{w(u,v)-1}$ such that the
bag of $N_i$ is $B \cup \{u^v_i, u^v_{i+1}, v^u_i, v^u_{i+1}\}$.
\item For each $(u,v) \in \edges{G}$, we take an arbitrary node whose bag $B$ 
  contains $u$ and add to its children a chain of nodes $N_1, \dots, 
  N_{w(u,v)}$ such that the bag of $N_i$ is $B \cup \{u^{v\square}_i\}$.
\item For each $(u,v) \in \edges{G}$, we take an arbitrary node whose bag $B$ 
  contains $v$ and add to its children a chain of nodes $N_1, \dots, 
  N_{w(u,v)}$ such that the bag of $N_i$ is $B \cup \{v^{u\square}_i\}$.
\item For each $v \in \vertices{G}$, we take an arbitrary node whose bag $B$ contains $v$
and add to its children
a chain of nodes
$N_1, \dots, N_{r-1}$
such that the bag of $N_i$ is
$B \cup \{h^v_i\}$.
\end{enumerate}
It is easy to verify that the result is a valid tree decomposition of the
primal graph of $I$ and its width is at most the treewidth of $G$ plus four.

It remains to show that our reduction is correct.
Obviously $I$ is a positive instance of \DAFNC\ iff it is a positive instance 
of \EDAFNC\ because the forbidden, necessary and complementary vertices make 
sure that every solution of the \DAFNC\ instance $I$ has exactly $k$ elements.
Hence we only consider \DAFNC.

The intention is that for each orientation of $G$ we have a solution
candidate $S$ in $I$ such that
an edge orientation from $u$ to $v$ entails
$V_{vu} \subseteq S$ and
$V_{uv} \cap S = \emptyset$,
and the other orientation entails $V_{uv} \subseteq S$ and $V_{vu} \cap S = \emptyset$.
For each vertex $v \in \vertices{G}$ and every incident edge $(v,u) \in 
\edges{G}$ regardless of its orientation, the vertex $v$ is attacked by the 
forbidden vertices $V_{vu}^\square$.
So every vertex $v \in \vertices{G}$ has as least as many attackers as the sum 
of the weights of all incident edges.
If in the orientation of $G$ all edges incident to $v$ are incoming edges, then 
each attack on $v$ from $V_{vu}^\square$ can be repelled by $V_{vu}$, since 
$V_{vu} \subseteq S$.
Due to the fact that the helper vertices $H_v$ consist of exactly $2r-1$ 
elements, $v$ can afford to have outgoing edges of total weight at most $r$.

We claim that $(G,w,r)$ is a positive instance of \MMO{} iff
$I$ is a positive instance of \DAFNC{}.

\medskip
\noindent
\emph{``Only if'' direction.}
Let $D$ be the directed graph given by an orientation of the edges of $G$ such that for each vertex the sum of weights of outgoing edges is at most $r$.
The set
$S = V_\triangle \cup \{v^u_1, \dots, v^u_{w(u,v)} \mid (u,v) \in \edges{D}\}$
is a defensive alliance in $G'$:
Let $x$ be an arbitrary element of $S$.
If $x$ is an element of a set $H_v$ or $V_{uv}$, then the only neighbor of $x$ 
in $G'$ is a necessary vertex, so $x$ can trivially defend itself; so suppose 
$x \in \vertices{G}$.
Let the sum of the weights of outgoing and incoming edges be denoted by 
$w^x_{\mathrm{out}}$ and $w^x_{\mathrm{in}}$, respectively.
The neighbors of $x$ that are also in $S$ consist of the elements of $H_x$ and 
all elements of sets $V_{xv}$ such that $(v,x) \in \edges{D}$.
Hence, including itself, $x$ has $2r + w^x_{\mathrm{in}}$ defenders in $G'$.
The attackers of $x$ consist of all elements of sets $V_{xv}$ such that $(x,v) 
\in \edges{D}$ (in total $w^x_{\mathrm{out}}$) and all elements of sets 
$V_{xv}^\square$ such that either $(v,x) \in \edges{D}$ or $(x,v) \in 
\edges{D}$ (in total $w^x_{\mathrm{in}} + w^x_{\mathrm{out}}$).
Hence $x$ has $w^x_{\mathrm{in}} + 2 w^x_{\mathrm{out}}$ attackers in $G'$.
This shows that $x$ has at least as many defenders as attackers, as by 
assumption $w^x_{\mathrm{out}} \leq r$.
Finally, it is easy to verify that
$\size{S} = k$,
$V_\square \cap S = \emptyset$,
$V_\triangle \subseteq S$,
and exactly one element of each pair of complementary vertices is in $S$.

\medskip
\noindent
\emph{``If'' direction.}
Let $S$ be a solution of $I$.
For every $(u,v) \in \edges{G}$, either $V_{uv} \subseteq S$ or $V_{vu} \subseteq S$ due to the complementary vertex pairs.
We define a directed graph $D$ by $\vertices{D} = \vertices{G}$ and
$\edges{D} = \{(u,v) \mid V_{vu} \subseteq S\} \cup \{(v,u) \mid V_{uv} \subseteq S\}$.
Suppose there is a vertex $x$ in $D$ whose sum of weights of outgoing edges is greater than $r$.
Clearly $x \in S$.
Let the sum of the weights of outgoing and incoming edges be denoted by 
$w^x_{\mathrm{out}}$ and $w^x_{\mathrm{in}}$, respectively.
The defenders of $x$ in $G'$ beside itself consist of the elements of $H_x$ and 
of $w^x_{\mathrm{in}}$ neighbors due to incoming edges in $D$.
These are in total $2r + w^x_{\mathrm{in}}$ defenders.
The attackers of $x$ in $G'$ consist of $2w^x_{\mathrm{out}}$ elements (of the 
form $x^v_i$ as well as $x^{v\square}_i$) due to outgoing edges in $D$ and 
$w^x_{\mathrm{in}}$ elements (of the form $x^{v\square}_i$) due to incoming 
edges.
These are in total $2w^x_{\mathrm{out}} + w^x_{\mathrm{in}}$ attackers.
But then $x$ has more attackers than defenders, as by assumption
$w^x_{\mathrm{out}} > r$.
\end{proof}

\subsection{Hardness of Defensive Alliance with Forbidden and Necessary 
Vertices}

Next we present a transformation $\tFNC$ that eliminates complementary vertex
pairs by turning a \DAFNC{} instance into an equivalent \DAFN{} instance.
Along with $\tFNC$, we define a function $\sFNC{I}$, for each \DAFNC{} instance
$I$,
such that the solutions of $I$ are in a one-to-one correspondence with those of
$\tFNC(I)$ in such a way that any two solutions of $I$ have the same size iff
the corresponding solutions of $\tFNC(I)$ have the same size.
We use these functions to obtain a polynomial-time reduction from \DAFNC{} to
\DAFN{} as well as from \EDAFNC{} to \EDAFN{}.

Before we formally define our reduction, we briefly describe the intuition 
behind the used gadgets.
The gadget in Figure~\ref{fig:complementary-reduction-gadget1} adds neighbors 
$a_1,\dots,a_n,a_1^\square,\dots,a_n^\square$ to every vertex $a$, which are so 
many that $a$ can only be in a solution if some of the new neighbors are also 
in the solution.
The new vertices are structured in such a way that every solution must in fact 
either contain all of $a,a_1,\dots,a_n$ or none of them.
Next, the gadget in Figure~\ref{fig:complementary-reduction-gadget2} is added 
for every complementary pair $(a,b)$.
This gadget is constructed in such a way that every solution must either 
contain all of $a_n,a^{ab},a^{ab}_1,\dots,a^{ab}_{n^2+n}$ or none of them, and 
the same holds for $b_n,b^{ab},b^{ab}_1,\dots,b^{ab}_{n^2+n}$.
By making the vertex $\triangle^{ab}$ necessary, every solution must contain 
one of these two sets.
At the same time, the bound on the solution size makes sure that we cannot 
afford to take both sets for any complementary pair.

\begin{definition}
We define a function \tFNC{}, which assigns a \DAFN{} instance to each \DAFNC{} 
instance
$I = (G,k,V_\square,V_\triangle,C)$.
For this, we use
$n$ to denote $\size{\vertices{G}}$
and first define a function
\[\sFNC{I} : x \mapsto x \cdot (n+1) + \size{C} \cdot (n^2+n+2).\]
For each $v \in \vertices{G}$, we introduce the following sets of new vertices.
\begin{align*}
  Y_v^\medcirc &{}=
  \{v_1, \dots, v_n\}
  &
  Y_v^\square &{}=
  \{v^\square_1, \dots, v^\square_n\}
\end{align*}
Next, for each $(a,b) \in C$, we introduce new vertices
$a^{ab}$, $b^{ab}$ and $\triangle^{ab}$ as well as,
for any $x \in \{a,b\}$, the following sets of new vertices.
\begin{align*}
  Z^{ab}_{x\medcirc} &{}=
  \{x^{ab}_1, \dots, x^{ab}_{n^2+n}\}
  &
  Z^{ab}_{x\square} &{}=
  \{x^{ab\square}_1, \dots, x^{ab\square}_{n^2+n}\}
\end{align*}
We use the notation
$u \oplus v$
to denote the set of edges
$\{(u,v),$ $(u,u^\square),$ $(v,v^\square),$ $(u,v^\square),$ $(v,u^\square)\}$.
\input{figures/complementary-reduction-gadget1}
\input{figures/complementary-reduction-gadget2}
Now we define
the \DAFN{} instance
$\tFNC(I) = (G',k',V_\square',V_\triangle')$, where
$k' = \sFNC{I}(k)$,
$V_\square' = V_\square \cup
\bigcup_{v \in \vertices{G}} Y_v^\square \cup
\bigcup_{(a,b) \in C} (Z^{ab}_{a\square} \cup Z^{ab}_{b\square})$,
$V_\triangle' = V_\triangle \cup \bigcup_{(a,b) \in C} \{\triangle^{ab}\}$
and $G'$ is the graph defined by
\begin{align*}
  \vertices{G'} = \vertices{G} &{}\cup
  \bigcup_{v \in \vertices{G}} (Y_v^\medcirc \cup Y_v^\square) \cup{}\\
  &{}\cup \bigcup_{(a,b) \in C} \big(\{\triangle^{ab}, a^{ab}, b^{ab}\} \cup 
  Z^{ab}_{a\medcirc} \cup Z^{ab}_{b\medcirc} \cup Z^{ab}_{a\square} \cup 
Z^{ab}_{b\square}\big),
\end{align*}
\begin{align*}
\edges{G'} = \edges{G} &{}\cup
\bigcup_{v \in \vertices{G}}
\big(
(\{v\} \times Y_v^\medcirc)
\cup
(\{v\} \times Y_v^\square)
\cup
\bigcup_{1\leq i < n} v_i \oplus v_{i+1}
\big)
\cup{}\\&{}\cup
\bigcup_{(a,b) \in C}
\bigcup_{x \in \{a,b\}} \big(
\{(\triangle^{ab}, x^{ab})\}
\cup
(\{x^{ab}\} \times Z^{ab}_{x\medcirc})
\cup{}\\
&\phantom{{}\cup \bigcup_{(a,b) \in C}\bigcup_{x \in \{a,b\}} \big(}
{}\cup
x_n \oplus x_1^{ab}
\cup
\bigcup_{1 \leq i < n^2+n}
x_i^{ab} \oplus x_{i+1}^{ab}
\big).
\end{align*}
We illustrate our construction in 
Figures~\ref{fig:complementary-reduction-gadget1} and 
\ref{fig:complementary-reduction-gadget2}.
\label{def:dafnc-to-dafn}
\end{definition}

\begin{lemma}
Let $I = (G,k,V_\square,V_\triangle,C)$ be a \DAFNC{} instance,
let $A$ be the set of solutions of $I$ and let $B$ be the set of solutions of 
the \DAFN{} instance $\tFNC(I)$.
There is a bijection $f: A \to B$ such that
$\size{f(S)} = \sFNC{I}(\size{S})$ holds for every $S \in A$.
\label{lem:dafnc-to-dafn-correct}
\end{lemma}

\begin{proof}
We use the same auxiliary notation as in Definition~\ref{def:dafnc-to-dafn}
and we define $f$ as
$S \mapsto S \cup
\bigcup_{v \in S} Y_v^\medcirc
\cup
\bigcup_{(a,b) \in C,\, x \in S \cap \{a,b\}} (\{\triangle^{ab}, x^{ab}\} \cup 
Z^{ab}_{x\medcirc})$.
For every $S \in A$,
we thus obtain
$\size{f(S)} = \sFNC{I}(\size{S})$,
and we first show that indeed $f(S) \in B$.

Let $S \in A$ and let $S'$ denote $f(S)$.
Obviously $S'$ satisfies
$V_\square' \cap S' = \emptyset$ and
$V_\triangle' \subseteq S'$.
To see that $S'$ is a defensive alliance in $G'$,
let $x$ be an arbitrary element of $S'$.
If $x \notin S$, then $x$ clearly has as least as many neighbors in $S'$ as 
neighbors not in $S'$ by construction of $f$, so suppose $x \in S$.
There is a defense
$\mu: N_G[x] \setminus S \to N_G[x] \cap S$
since $S$ is a defensive alliance in $G$.
We use this to construct a defense
$\mu': N_{G'}[x] \setminus S' \to N_{G'}[x] \cap S'$.
For any attacker $v$ of $x$ in $G'$, we distinguish two cases.
\begin{itemize}
\item If $v$ is some $x_i^\square \in Y_x^\square$ for some $x \in 
  \vertices{G}$, we set $\mu'(v) = x_i$.
This element is in $N_{G'}[x]$ by construction.
\item Otherwise $v$ is in $N_G[x] \setminus S$ (by our construction of $S'$).
Since the codomain of $\mu$ is a subset of the codomain of $\mu'$, we may set $\mu'(v) = \mu(v)$.
\end{itemize}
Since $\mu'$ is injective, each attack on $x$ in $G'$ can be repelled by $S'$.
Hence $S'$ is a defensive alliance in $G'$.

Clearly $f$ is injective.
It remains to show that $f$ is surjective.
Let $S'$ be a solution of $\tFNC(I)$.
First we make the following observations for each $v \in \vertices{G}$:
\begin{itemize}
\item
If $v \in S'$, then $Y_v^\medcirc \cap S' \neq \emptyset$
due to Observation~\ref{obs:one-in-s}, since
$Y_v^\medcirc \cup Y_v^\square$ contains a majority of neighbors of $v$, and 
the vertices in $Y_v^\square$ are forbidden.
\item
For each $v^{ab} \in S'$, where $(a,b) \in C$ such that $v=a$ or $v=b$, it 
holds that $Z^{ab}_{v\medcirc} \cap S' \neq \emptyset$ again due to 
Observation~\ref{obs:one-in-s}.
\item
If $S'$ contains an element of $Y_v^\medcirc$, then
$\{v\} \cup Y_v^\medcirc \cup \bigcup_{(v,z) \in C} Z^{vz}_{v\medcirc} \cup 
\bigcup_{(z,v) \in C} Z^{zv}_{v\medcirc} \subseteq S'$
by repeated applications of Observation~\ref{obs:all-in-s}.
To see this, note in particular that $N(v_n)$ can be partitioned into the two 
equal-sized sets
$\{v,v_{n-1}\} \cup \{v^{vz}_1 \mid (v,z) \in C\} \cup \{v^{zv}_1 \mid (z,v) 
\in C\}$
and
$\{v_{n-1}^\square, v_n^\square\} \cup \{v^{vz\square}_1 \mid (v,z) \in C\} 
\cup \{v^{zv\square}_1 \mid (z,v) \in C\}$,
and all vertices in the latter set are forbidden.
\item
If $S'$ contains an element of $Z^{ab}_{v\medcirc}$, where $(a,b) \in C$ such 
that $v=a$ or $v=b$, then $\{v^{ab}\} \cup Y_v^\medcirc  \cup \bigcup_{(v,z) 
\in C} Z^{vz}_{v\medcirc} \cup \bigcup_{(z,v) \in C} Z^{zv}_{v\medcirc} 
\subseteq S'$ for similar reasons.
\end{itemize}
It follows that for each $v \in \vertices{G}$, $S'$ contains either all or none 
of
$\{v\} \cup Y_v^\medcirc \cup \bigcup_{(v,z) \in C} \big(\{v^{vz}\} \cup 
Z^{vz}_{v\medcirc}\big) \cup \bigcup_{(z,v) \in C} \big(\{v^{zv}\} \cup 
Z^{zv}_{v\medcirc}\big)$.

For every $(a,b) \in C$,
$S'$ contains $a^{ab}$ or $b^{ab}$,
since $\triangle^{ab} \in S'$,
whose neighbors are $a^{ab}$ and $b^{ab}$.
It follows that $\size{S'} > \size{C} \cdot (n^2+n+2)$ even if $S'$ contains 
only one of each $(a,b) \in C$.
If, for some $(a,b) \in C$, $S'$ contained both $a$ and $b$, we could derive a 
contradiction to $\size{S'} \leq \sFNC{I}(k) = k \cdot (n+1) + \size{C} \cdot 
(n^2+n+2)$ because then $\size{S'} > (\size{C}+1) \cdot (n^2+n+2) >
\sFNC{I}(k)$.
So $S'$ contains either $a$ or $b$ for any $(a,b) \in C$.

We construct $S = S' \cap \vertices{G}$ and observe that
$S' = f(S)$,
$V_\triangle \subseteq S$, $V_\square \cap S = \emptyset$, and 
$\size{S \cap \{a,b\}} = 1$ for each $(a,b) \in C$.
It remains to show that $S$ is a defensive alliance in $G$.
Let $x$ be an arbitrary element of $S$.
We observe that $N_{G'}[x] \cap S' = (N_G[x] \cap S) \cup Y_x^\medcirc$ and 
similarly
$N_{G'}[x] \setminus S' = (N_G[x] \setminus S) \cup
Y_x^\square$.
Since the cardinality of each set $Y_x^\medcirc$ is equal to the cardinality of 
$Y_x^\square$,
this implies
$\size{N_{G'}[x] \cap S'} - \size{N_G[x] \cap S} = \size{N_{G'}[x] \setminus 
S'} - \size{N_G[x] \setminus S}$.
Since $S'$ is a defensive alliance in $G'$ and $x \in S'$,
it holds that $\size{N_{G'}[x] \cap S'} \geq \size{N_{G'}[x] \setminus S'}$.
We conclude
that
$\size{N_G[x] \cap S} \geq \size{N_G[x] \setminus S}$.
Hence $S$ is a defensive alliance in $G$.
\end{proof}

To obtain the hardness result for \DAFN{} parameterized by treewidth, it 
remains to show that the reduction specified by $\tFNC$ preserves bounded 
treewidth.

\begin{lemma}
\DAFN{}, parameterized by the treewidth of the graph, is \Wone{}-hard.
\label{lem:dafn-wone-hard}
\end{lemma}

\begin{proof}
Let $I$ be a \DAFNC{} instance whose primal graph we denote by $G$.
We obtain an equivalent \DAFN{} instance $\tFNC(I)$, whose graph we denote by 
$G'$.
This reduction is correct, as shown in Lemma~\ref{lem:dafnc-to-dafn-correct}.
It remains to show that the treewidth of $G'$ is bounded by a function of the treewidth of $G$.
Let $\calT$ be an optimal nice tree decomposition of $G$.
We build a tree decomposition $\calT'$ of $G'$ by modifying a copy of $\calT$ 
in the following way:
For each vertex $v \in \vertices{G}$, we add $v_n$ and $v_n^\square$ to every 
bag containing $v$.
Then we pick an arbitrary node $t$ in $\calT$ whose bag contains $v$, and we 
add new children $N_1,\dots,N_{n-1}$ to $t$ such that the bag of $N_i$ is 
$\{v,v_i,v_i^\square,v_{i+1},v_{i+1}^\square\}$.
Next, for every pair $(a,b)$ of complementary vertices, we pick an arbitrary 
node $t$ in $\calT$ whose bag $B$ contains both $a_n$ and $b_n$, and we add a 
chain of nodes $N_1, \dots, N_{2n^2+2n-1}$ between $t$ and its parent such 
that,
for $1 \leq i < n^2+n$, the bag of $N_i$ is $B \cup \{a^{ab}, a_i^{ab}, 
a_i^{ab\square}, a_{i+1}^{ab}, a_{i+1}^{ab\square}\}$,
the bag of $N_{n^2+n}$ is
$B \cup \{a^{ab}, b^{ab}, \triangle^{ab}\}$, and the bag of $N_{n^2+n+i}$ is $B 
\cup \{b^{ab}, b_{n^2+n+1-i}^{ab}, b_{n^2+n+1-i}^{ab\square}, b_{n^2+n-i}^{ab}, 
b_{n^2+n-i}^{ab\square}\}$.
It is easy to verify that $\calT'$ is a valid tree decomposition of $G'$.
Furthermore, the width of $\calT'$ is at most three times the width of $\calT$ 
plus five.
\end{proof}

The instances of \DAFNC{} are identical to the instances of the exact variant,
so $\tFNC$ is also applicable to the exact case.
In fact it turns out that this gives us also a reduction from \EDAFNC{} to 
\EDAFN{}.

\begin{lemma}
\EDAFN{}, parameterized by the treewidth of the graph, is \Wone{}-hard.
\label{lem:edafn-wone-hard}
\end{lemma}

\begin{proof}
Let $I$ and $I' = \tFNC(I)$ be our \EDAFNC{} and \EDAFN{} instances, 
respectively, and let
$k$ and $k'$ denote their respective solution sizes.
By Lemma~\ref{lem:dafnc-to-dafn-correct},
there is a bijection $f$ between the
solutions of $I$ and the solutions of $I'$
such that,
for every solution $S$ of $I$, $f(S)$ has $\sFNC{I}(k) = k'$ elements,
and for every solution $S'$ of $I'$,
$f^{-1}(S')$ has $k$ elements since $\sFNC{I}$ is invertible.
We can derive the bound on the treewidth of $I'$ as in the proof of 
Lemma~\ref{lem:dafn-wone-hard}.
\end{proof}

\subsection{Hardness of Defensive Alliance with Forbidden Vertices}

Now we present a transformation $\tFN$ that eliminates necessary vertices.
Our transformation not only operates on a problem instance, but also requires
an ordering $\preceq$ of the non-forbidden vertices of the graph.
Our reductions based on $\tFN$ will be correct for every ordering ${\preceq}$, 
but in order to keep the treewidth of the resulting problem instance bounded by 
the treewidth of the original problem instance, we must choose a suitable 
ordering.
We will describe this in detail later; for now we can consider ${\preceq}$ to 
be an arbitrary ordering of the non-forbidden vertices.

\input{figures/necessary-reduction-example-figure}
Before formally defining the transformation $\tFN$, we refer to
Figure~\ref{fig:necessary-reduction-example-figure}, which shows the result
for a simple example graph with only two vertices $a$ and $b$, of which $b$ is
necessary.
The basic idea is that the vertex $a'$ must be in every solution $S$:
If $a$ or any vertex to the left of $a$ is in $S$, it eventually forces $a'$ to 
be in $S$ as well.
Likewise, if $b$ or any vertex to the right of $b$ is in $S$, it also forces 
$a'$ to be in $S$.
Once $a' \in S$, the construction to the right of $a'$ makes sure that $b \in 
S$.
We will generalize this to instances containing more vertices so that every 
necessary vertex as well as the primed copy of each non-necessary vertex is in 
every solution.

\begin{definition}
We define a function \tFN{}, which assigns a \DAF{} instance to each pair $(I,{\preceq})$,
where
$I = (G,k,V_\square,V_\triangle)$ is a \DAFN{} instance and
${\preceq}$ is an ordering of the non-forbidden elements of $\vertices{G}$.
For this,
let $V_\medcirc$ denote $\vertices{G} \setminus (V_\square \cup V_\triangle)$.
We use
$n$ to denote $\size{\vertices{G}}$,
and we first define a function
$\sFN{I} : x \mapsto (n+3) \cdot (x + \size{V_\medcirc}) - \size{V_\triangle}$.
We use $H$ to denote the set of new vertices $\{v', g_v, h_v, g_v^\square, 
h_v^\square \mid v \in V_\medcirc\}$.
The intention is for each $g_v^\square$ and $h_v^\square$ to be forbidden, for 
each $v'$ and $h_v$ to be in every solution, and for $g_v$ to be in a solution 
iff $v$ is in it at the same time.
We write
$V^+$ to denote $V_\triangle \cup V_\medcirc \cup \{v' \mid v \in 
V_\medcirc\}$;
for each $v \in V^+$, we use $A_v$ to denote the set of new vertices
$\{v_1, \dots, v_{n+1}, v_1^\square, \dots, v_{n+1}^\square \}$,
and we use shorthand notation $A_v^\medcirc = \{v_1, \dots, v_{n+1}\}$ and 
$A_v^\square = \{v_1^\square, \dots, v_{n+1}^\square\}$.
The intention is for each $v_i^\square$ to be forbidden and for each $v_i$ to 
be in a solution iff $v$ is in it at the same time.
We use the notation 
$u \oplus v$
to denote the set of edges
$\{(u,v), (u,u^\square), (v,v^\square), (u,v^\square), (v,u^\square)\}$.
For any vertex $v \in V_\medcirc \cup V_\triangle$, we define $p(v) = v$ if $v 
\in V_\triangle$ and $p(v) = v'$ if $v \in V_\medcirc$.
Let $P$ be the set consisting of all pairs $(p(u),p(v))$ such that $v$ is the 
direct successor of $u$ according to ${\preceq}$.
Now we define $\tFN(I,{\preceq}) = (G',k',V_\square')$, where
$V'_\square = V_\square \cup \{g_v^\square, h_v^\square \mid v \in V_\medcirc\} 
\cup \bigcup_{v \in V^+} A_v^\square$,
$k' = \sFN{I}(k)$, and
$G'$ is the graph defined by
\begin{align*}
\vertices{G'} ={}& \vertices{G} \cup H \cup \bigcup_{v \in V^+} A_v,\\
\edges{G'} ={}& \edges{G} \cup \{ (v, v_i), (v,v_i^\square)
  \mid v \in V^+,\; 1 \leq i \leq n+1 \}\\
  {}\cup{}& \bigcup_{v \in V^+,\; 1 \leq i \leq n} v_i \oplus v_{i+1}
  \cup \bigcup_{(u,v) \in P} u_{n+1} \oplus v_1\\
  {}\cup{}& \bigcup_{v \in V_\medcirc} v_{n+1} \oplus g_v
  \cup \{ (v',g_v), (v',h_v), (g_v,h_v), (g_v,h_v^\square) \mid v \in 
  V_\medcirc \}.
\end{align*}
We illustrate our construction in Figure~\ref{fig:necessary-reduction-gadget1} and~\ref{fig:necessary-reduction-gadget2}.
\label{def:dafn-to-daf}
\end{definition}

\input{figures/necessary-reduction-gadget1}
\input{figures/necessary-reduction-gadget2}

We now prove that $\tFN$ yields a correct reduction for any ordering $\preceq$.

\begin{lemma}
Let $I = (G,k,V_\square,V_\triangle)$ be a \DAFN{} instance,
let ${\preceq}$ be an ordering of $\vertices{G} \setminus V_\square$,
let $A$ be the set of solutions of $I$ and
let $B$ be the set of solutions of the \DAF{} instance $\tFN(I,{\preceq})$.
There is a bijection $f: A \to B$ such that
$\size{f(S)} = \sFN{I}(\size{S})$
holds for every $S \in A$.
\label{lem:dafn-to-daf-correct}
\end{lemma}

\begin{proof}
We use the same auxiliary notation as in Definition~\ref{def:dafn-to-daf}
and we define $f$ as
\[f(S) = S
\cup \bigcup_{v \in S} A_v^\medcirc
\cup \{v', h_v \mid v \in V_\medcirc\}
\cup \bigcup_{v \in V_\medcirc} A_{v'}^\medcirc
\cup \{g_v \mid v \in S \cap V_\medcirc\}
.\]
For every $S \in A$,
we thus obtain
$\size{f(S)}
= \size{S} + \size{S} (n+1) + 2 \size{V_\medcirc} + \size{V_\medcirc} \cdot 
(n+1) + (\size{S}-\size{V_\triangle})
= \sFN{I}(\size{S})$,
and we first show that indeed $f(S) \in B$.

Let $S \in A$ and let
$S'$ denote $f(S)$.
Obviously $S'$ satisfies
$V_\square' \cap S' = \emptyset$.
To see that $S'$ is a defensive alliance in $G'$,
let $x$ be an arbitrary element of $S'$.
If $x \in S$, then there is a defense
$\mu: N_G[x] \setminus S \to N_G[x] \cap S$
since $S$ is a defensive alliance in $G$.
We use this to construct a defense
$\mu': N_{G'}[x] \setminus S' \to N_{G'}[x] \cap S'$.
For any attacker $a$ of $x$ in $G'$, we distinguish the following cases:
\begin{itemize}
\item If $a$ is some $v_i^\square \in A_v^\square$ for some $v \in V^+$,
then $x$ is either $v_i$ or a neighbor of $v_i$, all of which are in $S'$,
and we set $\mu'(a) = v_i$.
\item Similarly, if $a$ is $g_v^\square$ for some $v \in V_\medcirc$,
then we set $\mu'(a) = g_v$.
\item If $a$ is $h_v^\square$ for some $v \in V_\medcirc$,
then $x = g_v$ and we set $\mu'(a) = h_v$.
\item If $a$ is $g_v$ for some $v \in V_\medcirc$,
then $x$ is either $v'$ or $h_v$,
which is not used for repelling any other attack because $h_v^\square$ cannot 
attack $x$,
so we set $\mu'(a) = h_v$.
\item Otherwise $a$ is in $N_G[x] \setminus S$ (by our construction of $S'$).
Since the codomain of $\mu$ is a subset of the codomain of $\mu'$, we may set $\mu'(a) = \mu(a)$.
\end{itemize}
Since $\mu'$ is injective, each attack on $x$ in $G'$ can be repelled by $S'$.
Hence $S'$ is a defensive alliance in $G'$.

Clearly $f$ is injective.
It remains to show that $f$ is surjective.
Let $S'$ be a solution of $\tFN(I,{\preceq})$.
We first show that $V_\triangle \cup \{v', h_v \mid v \in V_\medcirc\} 
\subseteq S'$:
\begin{itemize}
\item If $S'$ contains some $v \in V^+$, then
$S'$ contains an element of $A_v^\medcirc$
by Observation~\ref{obs:one-in-s}.

\item If $S'$ contains an element of $A_v^\medcirc$ for some $v \in V^+$, then
$\{v\} \cup A_v^\medcirc \subseteq S'$
by Observation~\ref{obs:all-in-s}.

\item If $v_{n+1} \in S'$ for some $v \in V_\medcirc$, then
$g_v \in S'$ for the same reason.

\item Furthermore, if $S'$ contains an element of $A_v^\medcirc$ for some $v 
  \in V_\triangle \cup \{v' \mid v \in V_\medcirc\}$, then
also $A_{u}^\medcirc \subseteq S'$ for every $u \in V_\triangle \cup \{v' \mid 
v \in V_\medcirc\}$
for the same reason.

\item If $g_v \in S'$ for some $v \in V_\medcirc$, then
$\{h_v,v',v_{n+1}\} \subseteq S'$ by Observation~\ref{obs:all-in-s}.

\item If $h_v \in S'$ for some $v \in V_\medcirc$, then
$a' \in S'$ because at least $g_v$ or $v'$ must be in $S'$
and the former implies $v' \in S'$ as we have seen.

\item
Since $S'$ is nonempty, the previous observations show that
for every $v \in V_\triangle \cup \{v' \mid v \in V_\medcirc\}$
it holds that
$\{v\} \cup A_v^\medcirc \subseteq S'$.
Finally, we show that $\{h_v \mid v \in V_\medcirc\} \subseteq S'$.
Suppose, for the sake of contradiction, that there is some
$v \in V_\medcirc$ such that
$h_v \notin S'$.
We have seen that the latter can only be the case if $g_v \notin S'$, and we 
know that $v' \in S'$.
We obtain the contradiction that $v'$ is attacked by $g_v$, $h_v$ and 
$A_{v'}^\square$, whereas its only defenders are $v'$ itself and 
$A_{v'}^\medcirc$.
\end{itemize}
Let $S = S' \cap \vertices{G}$.
By the previous observations, it is easy to see that $S' = f(S)$.
It remains to show that $S$ is a defensive alliance in $G$.
Let $x$ be an arbitrary element of $S$.
We observe that $N_{G'}[x] \cap S' = (N_G[x] \cap S) \cup
A_x^\medcirc$ and similarly
$N_{G'}[x] \setminus S' = (N_G[x] \setminus S) \cup
A_v^\square$.
Since $\size{A_x^\medcirc} = \size{A_x^\square}$, this implies
$\size{N_{G'}[x] \cap S'} - \size{N_G[x] \cap S} = \size{N_{G'}[x] \setminus 
S'} - \size{N_G[x] \setminus S}$.
Since $S'$ is a defensive alliance in $G'$ and $x \in S'$,
it holds that $\size{N_{G'}[x] \cap S'} \geq \size{N_{G'}[x] \setminus S'}$.
We conclude
that
$\size{N_G[x] \cap S} \geq \size{N_G[x] \setminus S}$.
Hence $S$ is a defensive alliance in $G$.
\end{proof}

Given an ordering ${\preceq}$, clearly $\tFN(I,\preceq)$ is computable in 
polynomial time.
We can thus easily obtain a reduction from \DAFN{} to \DAF{}
by first computing an arbitrary ordering ${\preceq}$ of the non-forbidden 
vertices.
%
We next show that by choosing $\preceq$ appropriately, this amounts to an FPT 
reduction that preserves bounded treewidth.

\begin{lemma}
\DAF{}, parameterized by the treewidth of the graph, is \Wone{}-hard.
\label{lem:daf-wone-hard}
\end{lemma}

\begin{proof}
Let $I = (G,k,V_\square,V_\triangle)$ be a \DAFN{} instance
and let $\calT$ be an optimal nice tree decomposition of $G$.
We can compute such a tree decomposition in
FPT time~\cite{siamcomp:Bodlaender96}.
Let ${\preceq}$ be 
the
ordering of the elements of $V_\triangle \cup V_\medcirc$
that is obtained in linear time by doing a post-order traversal of $\calT$ and 
sequentially recording the elements that occur for the last time in the current 
bag.
We obtain the \DAF{} instance $\tFN(I,\preceq)$, whose graph we denote by $G'$.
This reduction is correct, as shown in Lemma~\ref{lem:dafn-to-daf-correct},
and computable in FPT time.
It remains to show that the treewidth of $G'$ is bounded by a function of the treewidth of $G$.
To this end, we use $\calT$ to build a tree decomposition $\calT'$ of $G'$.
We initially set $\calT' := \calT$ and modify it by the following steps:
\begin{enumerate}
  \item For each $v \in V_\medcirc$, we add $g_v$, $g_v^\square$, $h_v$, 
    $h_v^\square$ and $v'$ to the bag of $t_v^{\calT'}$.
    Note that afterwards $t_v^{\calT'} = t_{v'}^{\calT'}$.
    After this step we increased the width of $\calT'$ by at most five.

  \item For each $v \in V^+$, we use $B_v$ to denote the bag of $t_v^{\calT'}$ 
    and replace $t_v^{\calT'}$ by a chain of nodes $N_1, \dots, N_n$, where 
    $N_n$ is the topmost node and the bag of $N_i$ is $B_v \cup \{v_i, 
    v_i^\square, v_{i+1}, v_{i+1}^\square\}$.
    After this step we increased the width of $\calT'$ by at most nine.
    Note that the bag of the new node $t_v^{\calT'}$ now contains $v_{i+1}$ and 
    $v_{i+1}^\square$.
    We have so far covered all edges except the ones connecting elements of two 
    different sets $A_x$ and $A_y$ for $(x,y) \in P$.

  \item For every $(u,v) \in P$, we add $v_1$ and $v_1^\square$ into the bag of 
    every node between (and including) $t_u^{\calT'}$ and $t_{v_1}^{\calT'}$.
    Note that this preserves connectedness and afterwards the bag of 
    $t_u^{\calT'}$ contains $u_{i+1}$, $u_{i+1}^\square$, $v_1$ and 
    $v_1^\square$, thus covering the remaining edges.
    After this step we increased the width of $\calT'$ by at most 13.
    (Since the number of children of each tree decomposition node is at most 
    two, this step enlarges every bag at most twice.)
\end{enumerate}
It is easy to verify that $\calT'$ is a valid tree decomposition of $G'$.
Furthermore, the width of $\calT'$ is at most the width of $\calT$ plus 13.
\end{proof}

We again get an analogous result for the exact variant.

\begin{corollary}
\EDAF{}, parameterized by the treewidth of the graph, is \Wone{}-hard.
\label{cor:edaf-wone-hard}
\end{corollary}

\subsection{Hardness of Defensive Alliance}
\label{sec:da-wone-tw}

We now introduce a transformation $\tF$ that eliminates forbidden vertices.
The basic idea is that we ensure that a forbidden vertex $f$ is never part of a 
solution by adding so many neighbors to $f$ that we could only defend $f$ by 
exceeding the bound on the solution size.

\begin{definition}
We define a function \tF{}, which assigns a \DA{} instance to each \DAF{} 
instance
$I = (G,k,V_\square)$.
\label{def:daf-to-da}
For each $f \in V_\square$, we
introduce new vertices $f', f_1, \dots, f_{2k}$.
Now we define $\tF(I) = (G',k)$, where
$G'$ is the graph defined by
\begin{align*}
\vertices{G'} ={}& \vertices{G} \cup \{f', f_1, \dots, f_{2k} \mid f \in V_\square\},\\
\edges{G'} ={}& \edges{G} \cup \{ (f,f_i),\; (f',f_i) \mid f \in V_\square,\; 1 \leq i \leq 2k\}.
\end{align*}
\end{definition}

We now prove that $\tF$ yields a correct reduction from \DAF{} to \DA{}.

\begin{lemma}
Every \DAF{} instance $I$ has the same solutions as the \DA{} instance 
$\tF(I)$.
\label{lem:daf-to-da-correct}
\end{lemma}

\begin{proof}
Let $I = (G,k,V_\square)$ and $\tF(I) = (G',k)$.
Each solution $S$ of $I$ is also a solution of $\tF(I)$ because the subgraph of 
$G$ induced by $N_G[S]$ is equal to the subgraph of $G'$ induced by 
$N_{G'}[S]$.
Now let $S'$ be a solution of $\tF(I)$.
For every $f \in V_\square$,
neither $f$ nor $f'$ are in $S'$
because each of these vertices has at least $2k$ neighbors,
and $S'$ cannot contain any $f_i$ because $N_{G'}(f_i) = \{f,f'\}$.
Hence $S'$ is also a solution of $I$ as the subgraphs induced by the respective 
neighborhoods are again equal.
\end{proof}

We use $\tF$ to show \Wone{}-hardness of \DA{} by reducing from \DAF{} while 
preserving bounded treewidth.

\begin{lemma}
\DA{}, parameterized by the treewidth of the graph, is \Wone{}-hard.
\label{lem:da-wone-hard}
\end{lemma}

\begin{proof}
Let $I = (G,k,V_\square)$ be a \DAF{} instance,
let $G'$ denote the graph of $\tF(I)$
and let $\calT$ be an optimal nice tree decomposition of $G$.
We build a tree decomposition $\calT'$ of $G'$ by modifying a copy of $\calT$ in the following way:
For every $f \in V_\square$, we pick an arbitrary node $t$ in $\calT$ whose bag $B$ contains $f$,
and we add a chain of nodes $N_1, \dots, N_{2k}$ between $t$ and its parent such that,
for $1 \leq i \leq 2k$, the bag of $N_i$ is
$B \cup \{f', f_i\}$.
It is easy to verify that $\calT'$ is a valid tree decomposition of $G'$.
Furthermore, the width of $\calT'$ is at most the width of $\calT$ plus two.
\end{proof}

We again get an analogous result for the exact variant.

\begin{corollary}
\EDA{}, parameterized by the treewidth of the input graph, is \Wone{}-hard.
\label{cor:eda-wone-hard}
\end{corollary}

\section{Conclusion}
\label{sec:conclusion}

In this work, we proved that the problem of deciding whether a given graph 
possesses a nonempty defensive alliance whose size is at most a given integer 
is \Wone-hard when the parameter is the treewidth of the graph.
This means that no fixed-parameter tractable algorithm exists under the common 
complexity-theoretic assumption $\Wone \neq \FPT$.
Still, recent work has shown the problem to be solvable in polynomial time on 
graphs of bounded clique-width~\cite{dam:KiyomiO17}, which implies that there 
is a polynomial-time algorithm for graphs of bounded treewidth.
Our result proves that, for any such algorithm, the degree of this polynomial 
must necessarily depend on the treewidth unless $\Wone = \FPT$.
For future research it may be interesting to study related problems 
corresponding to other alliance notions such as offensive alliances.


\bibliography{references}

\end{document}

%% file: figures/da-example-figure.tex
\begin{figure}
\centering
\begin{tikzpicture}
\node [draw,circle] (a) at (0.5,0) {$a$};
\node [draw,circle] (b) at (1.5,0) {$b$};
\node (c) at (0,-1) {$c$};
\node (d) at (1,-1) {$d$};
\node (e) at (2,-1) {$e$};

\draw (a) -- (b);
\draw (c) -- (a) -- (d) -- (b) -- (e);
\draw (c) -- (d) -- (e);
\draw [bend right] (c) to (e);
\end{tikzpicture}
\caption{A graph with a minimum non-empty defensive alliance indicated by circled vertices}
\label{fig:da-example}
\end{figure}
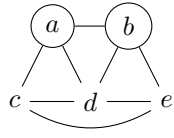

%% file: figures/da-example-figure2.tex
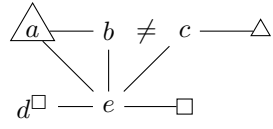
\begin{figure}
\centering
\begin{tikzpicture}
\node [draw,necessary] (a) at (0,0) {$a$};
\node (b) at (1,0) {$b$};
\node at (1.5,0) {$\neq$};
\node (c) at (2,0) {$c$};
\node [necessary] (g) at (3,0) {};
\node (d) at (0,-1) {$d^\square$};
\node (e) at (1,-1) {$e$};
\node [forbidden] (f) at (2,-1) {};

\draw (a) -- (b) -- (e) -- (a);
\draw (d) -- (e) -- (f);
\draw (e) -- (c);
\draw (c) -- (g);
\end{tikzpicture}
\caption{Illustration of forbidden, necessary and complementary vertices}
\label{fig:da-example2}
\end{figure}

%% file: figures/td-example-figure.tex
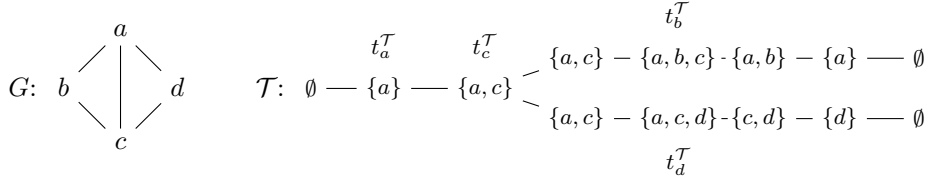
\begin{figure}[t]%
\centering
$G$:
\begin{tikzpicture}[scale=0.75, baseline=(d.base)]
\node (a) at (0,1) {$a$};
\node (b) at (-1,0) {$b$};
\node (c) at (0,-1) {$c$};
\node (d) at (1,0) {$d$};
\draw (a) -- (b) -- (c) -- (d) -- (a) -- (c);
\end{tikzpicture}%
\hspace{2em}
$\calT$:
\begin{tikzpicture}[level/.style={sibling distance=8mm,level distance=12mm},
level 1/.style={level distance=8mm},
grow'=right, font=\footnotesize, anchor=west, growth parent anchor=west,
baseline=(root.base)]
\node (root) {$\emptyset$}
	child {node [label=above:$t^\calT_a$] {$\{a\}$}
		child {node [label=above:$t^\calT_c$] {$\{a,c\}$}
			child {node {$\{a,c\}$}
				child {node [label=above:$t^\calT_b$] {$\{a,b,c\}$}
					child {node {$\{a,b\}$}
						child {node {$\{a\}$}
							child {node {$\emptyset$}}
						}
					}
				}
			}
			child {node {$\{a,c\}$}
				child {node [label=below:$t^\calT_d$] {$\{a,c,d\}$}
					child {node {$\{c,d\}$}
						child {node {$\{d\}$}
							child {node {$\emptyset$}}
						}
					}
				}
			}
		}
	}
	;
\end{tikzpicture}%
\caption{A graph $G$ and a nice tree decomposition $\calT$ of $G$ rooted at the leftmost node}
\label{fig:td-example}
\end{figure}

%% file: figures/minmaxoutdegree-reduction-example.tex
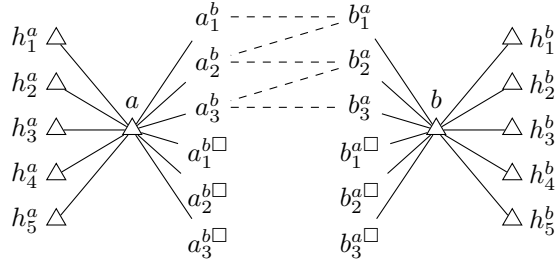
\begin{figure}
\centering
\begin{tikzpicture}[yscale=0.6]
\node (a) at (-2,0) [necessary,label=$a$] {};
\node (b) at (2,0) [necessary,label=$b$] {};

\node (ab1) at (-1,2.5) {$a^b_1$};
\node (ab2) at (-1,1.5) {$a^b_2$};
\node (ab3) at (-1,0.5) {$a^b_3$};
\node (ab1s) at (-1,-0.5) {$a^{b\square}_1$};
\node (ab2s) at (-1,-1.5) {$a^{b\square}_2$};
\node (ab3s) at (-1,-2.5) {$a^{b\square}_3$};

\node (ba1) at (1,2.5) {$b^a_1$};
\node (ba2) at (1,1.5) {$b^a_2$};
\node (ba3) at (1,0.5) {$b^a_3$};
\node (ba1s) at (1,-0.5) {$b^{a\square}_1$};
\node (ba2s) at (1,-1.5) {$b^{a\square}_2$};
\node (ba3s) at (1,-2.5) {$b^{a\square}_3$};

\node (ha1) at (-3,2) [necessary,label=left:$h^a_1$] {};
\node (ha2) at (-3,1) [necessary,label=left:$h^a_2$] {};
\node (ha3) at (-3,0) [necessary,label=left:$h^a_3$] {};
\node (ha4) at (-3,-1) [necessary,label=left:$h^a_4$] {};
\node (ha5) at (-3,-2) [necessary,label=left:$h^a_5$] {};
\node (hb1) at (3,2) [necessary,label=right:$h^b_1$] {};
\node (hb2) at (3,1) [necessary,label=right:$h^b_2$] {};
\node (hb3) at (3,0) [necessary,label=right:$h^b_3$] {};
\node (hb4) at (3,-1) [necessary,label=right:$h^b_4$] {};
\node (hb5) at (3,-2) [necessary,label=right:$h^b_5$] {};

\foreach \v in {a,b}
  \foreach \i in {1,...,5}
    \draw (\v) -- (h\v\i);

\foreach \i in {1,...,3} {
  \draw (a) -- (ab\i);
  \draw (a) -- (ab\i s);
  \draw (b) -- (ba\i);
  \draw (b) -- (ba\i s);
}

\draw [dashed] (ab1) to (ba1);
\draw [dashed] (ab2) to (ba2);
\draw [dashed] (ab3) to (ba3);

\draw [dashed] (ba1) -- (ab2);
\draw [dashed] (ba2) -- (ab3);
\end{tikzpicture}
\caption{%
Result of our transformation on a sample \MMO\ instance with
$r=3$ and two vertices $a,b$ that are connected by an edge of weight~$3$.
Complementary vertex pairs are shown via dashed lines.
Necessary and forbidden vertices have a $\triangle$ and $\square$ symbol next
to their name, respectively.%
}
\label{fig:minmaxoutdegree-reduction-example}
\end{figure}

%% file: figures/complementary-reduction-gadget1.tex
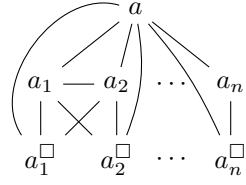
\begin{figure}
\centering
\begin{tikzpicture}
\node (a) at (1.25,0) {$a$};

\node (a1) at (0,-1) {$a_1$};
\node (a2) at (1,-1) {$a_2$};
\node at (1.75,-1) {$\cdots$};
\node (an) at (2.5,-1) {$a_n$};

\node (a1s) at (0,-2) {$a_1^{\square}$};
\node (a2s) at (1,-2) {$a_2^{\square}$};
\node at (1.75,-2) {$\cdots$};
\node (ans) at (2.5,-2) {$a_n^{\square}$};

\draw (a) -- (a1);
\draw (a) -- (a2);
\draw (a) -- (an);
\draw (a1) -- (a2);

\draw (an) -- (ans);

\draw (a1) -- (a1s) -- (a2) -- (a2s);
\draw (a1) -- (a2s);

\draw [bend left=70] (a1s) to (a);
\draw [bend right=15] (a2s) to (a);
\draw [bend right=8] (ans) to (a);
\end{tikzpicture}
\caption{Gadget for each vertex $a$ of the original graph in the reduction from \SSFNC{} to \SSFN{}. The vertex $a$ may have additional neighbors from the original graph, and the vertices $a_n$ and $a_n^\square$ may have additional neighbors as depicted in Figure~\ref{fig:complementary-reduction-gadget2}.}
\label{fig:complementary-reduction-gadget1}
\end{figure}

%% file: figures/complementary-reduction-gadget2.tex
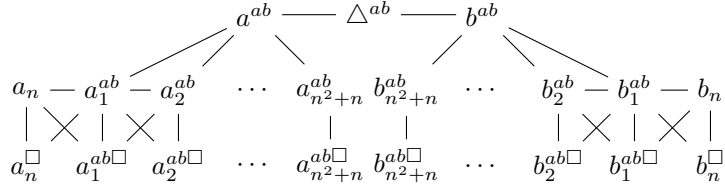
\begin{figure}
\centering
\begin{tikzpicture}
\node (ap) at (5.5,0) {$a^{ab}$};

\node (an) at (2.5,-1) {$a_n$};
\node (a1ab) at (3.5,-1) {$a_1^{ab}$};
\node (a2ab) at (4.5,-1) {$a_2^{ab}$};
\node at (5.5,-1) {$\cdots$};
\node (a1ab') at (6.5,-1) {$a_{n^2+n}^{ab}$};

\node (ans) at (2.5,-2) {$a_n^\square$};
\node (a1abs) at (3.5,-2) {$a_1^{ab\square}$};
\node (a2abs) at (4.5,-2) {$a_2^{ab\square}$};
\node at (5.5,-2) {$\cdots$};
\node (a1abs') at (6.5,-2) {$a_{n^2+n}^{ab\square}$};

\node (triangle) at (7,0) {$\triangle^{ab}$};

\node (bp) at (8.5,0) {$b^{ab}$};

\node (b1ab') at (7.5,-1) {$b_{n^2+n}^{ab}$};
\node at (8.5,-1) {$\cdots$};
\node (b2ab) at (9.5,-1) {$b_2^{ab}$};
\node (b1ab) at (10.5,-1) {$b_1^{ab}$};
\node (bn) at (11.5,-1) {$b_n$};

\node (b1abs') at (7.5,-2) {$b_{n^2+n}^{ab\square}$};
\node at (8.5,-2) {$\dots$};
\node (b2abs) at (9.5,-2) {$b_2^{ab\square}$};
\node (b1abs) at (10.5,-2) {$b_1^{ab\square}$};
\node (bns) at (11.5,-2) {$b_n^{\square}$};

\draw (ap) -- (a1ab);
\draw (ap) -- (a2ab);
\draw (ap) -- (a1ab');
\draw (ap) -- (triangle) -- (bp);
\draw (bp) -- (b1ab);
\draw (bp) -- (b2ab);
\draw (bp) -- (b1ab');

\draw (an) -- (a1ab) -- (a2ab);
\draw (bn) -- (b1ab) -- (b2ab);

\draw (an) -- (ans);
\draw (a1ab) -- (a1abs);
\draw (a2ab) -- (a2abs);
\draw (a1ab') -- (a1abs');
\draw (bn) -- (bns);
\draw (b1ab) -- (b1abs);
\draw (b2ab) -- (b2abs);
\draw (b1ab') -- (b1abs');

\draw (an) -- (a1abs) -- (a2ab);
\draw (ans) -- (a1ab) -- (a2abs);
\draw (bn) -- (b1abs) -- (b2ab);
\draw (bns) -- (b1ab) -- (b2abs);
\end{tikzpicture}
\caption{Gadget for each pair of complementary vertices $(a,b)$ in the reduction from \SSFNC{} to \SSFN{}. The vertices $a_n$, $a_n^\square$, $b_n$ and $b_n^\square$ have additional neighbors as depicted in Figure~\ref{fig:complementary-reduction-gadget1}.}
\label{fig:complementary-reduction-gadget2}
\end{figure}

%% file: figures/necessary-reduction-example-figure.tex
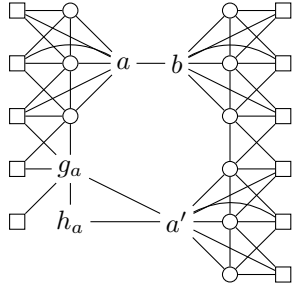
\begin{figure}
\centering
\begin{tikzpicture}[scale=0.7,inner sep=2pt]
\node (a) at (2,4) {$a$};

\node (a1) at (1,5) [anonymous] {};
\node (a2) at (1,4) [anonymous] {};
\node (a3) at (1,3) [anonymous] {};

\node (a1s) at (0,5) [forbidden] {};
\node (a2s) at (0,4) [forbidden] {};
\node (a3s) at (0,3) [forbidden] {};

\node (ga) at (1,2) {$g_a$};
\node (ha) at (1,1) {$h_a$};
\node (gas) at (0,2) [forbidden] {};
\node (has) at (0,1) [forbidden] {};

\node (ap) at (3,1) {$a'$};
\node (ap1) at (4,2) [anonymous] {};
\node (ap2) at (4,1) [anonymous] {};
\node (ap3) at (4,0) [anonymous] {};
\node (ap1s) at (5,2) [forbidden] {};
\node (ap2s) at (5,1) [forbidden] {};
\node (ap3s) at (5,0) [forbidden] {};

\node (b) at (3,4) {$b$};

\node (b1) at (4,5) [anonymous] {};
\node (b2) at (4,4) [anonymous] {};
\node (b3) at (4,3) [anonymous] {};

\node (b1s) at (5,5) [forbidden] {};
\node (b2s) at (5,4) [forbidden] {};
\node (b3s) at (5,3) [forbidden] {};

\draw (a) -- (b);

\draw (a) -- (a1);
\draw (a) -- (a2);
\draw (a) -- (a3);

\draw (ap) -- (ga);
\draw (ap) -- (ha);

\draw (a1) -- (a1s);
\draw (a2) -- (a2s);
\draw (a3) -- (a3s);
\draw (ga) -- (gas);

\draw (a1) -- (a2) -- (a3) -- (ga) -- (ha);
\draw (a1) -- (a2s) -- (a3) -- (gas);
\draw (a1s) -- (a2) -- (a3s) -- (ga);
\draw (has) -- (ga);

\draw (b) -- (b1);
\draw (b) -- (b2);
\draw (b) -- (b3);

\draw (b1) -- (b1s);
\draw (b2) -- (b2s);
\draw (b3) -- (b3s);

\draw (ap) -- (ap1) -- (ap1s);
\draw (ap) -- (ap2) -- (ap2s);
\draw (ap) -- (ap3) -- (ap3s);

\draw (b1) -- (b2) -- (b3) -- (ap1) -- (ap2) -- (ap3);
\draw (b1) -- (b2s) -- (b3) -- (ap1s) -- (ap2) -- (ap3s);
\draw (b1s) -- (b2) -- (b3s) -- (ap1) -- (ap2s) -- (ap3);

\draw (ap1s) to (ap);
\draw [bend right] (ap2s) to (ap);
\draw (ap3s) to (ap);

\draw (a1s) -- (a);
\draw [bend left] (a2s) to (a);
\draw (a3s) -- (a);
\draw (b1s) -- (b);
\draw [bend right] (b2s) to (b);
\draw (b3s) -- (b);
\end{tikzpicture}
\caption{Result of the transformation $\tFN{}$ applied to an example graph with two adjacent vertices $a$ and $b$, where $b$ is necessary. Every solution in the depicted graph contains $a'$, $h_a$ and $b$.}
\label{fig:necessary-reduction-example-figure}
\end{figure}

%% file: figures/necessary-reduction-gadget1.tex
\begin{figure}
\centering
\begin{tikzpicture}
[xscale=1.5,inner sep=2pt,bend angle=20]
\node (a) at (2,2) {$a$};

\node (a1) at (0,1) {$a_1$};
\node (a2) at (1,1) {$a_2$};
\node (dots) at (2,1) {$\cdots$};
\node (an) at (3,1) {$a_n$};
\node (an1) at (4,1) {$a_{n+1}$};

\node (a1s) at (0,0) {$a_1^\square$};
\node (a2s) at (1,0) {$a_2^\square$};
\node (ans) at (3,0) {$a_n^\square$};
\node (an1s) at (4,0) {$a_{n+1}^\square$};

\node (ga) at (5,1) {$g_a$};
\node (gas) at (5,0) {$g_a^\square$};
\node (has) at (6,0) {$h_a^\square$};
\node (ha) at (6,1) {$h_a$};
\node (ap) at (5.5,2) {$a'$};

\draw [bend right] (a) to (a1);
\draw (a) -- (a2);
\draw (a) -- (an);
\draw [bend left] (a) to (an1);

\draw (ap) -- (ga);
\draw (ap) -- (ha);

\draw (a1) -- (a1s);
\draw (a2) -- (a2s);
\draw (an) -- (ans);
\draw (an1) -- (an1s);
\draw (ga) -- (gas);

\draw (a1) -- (a2);
\draw (an) -- (an1) -- (ga) -- (ha);
\draw (a1) -- (a2s);
\draw (a2) -- (a1s);
\draw (an) -- (an1s);
\draw (an1) -- (ans);
\draw (an1) -- (gas);
\draw (ga) -- (an1s);
\draw (has) -- (ga);

\draw [bend left] (a1s) to (a);
\draw (a2s) -- (a);
\draw (ans) -- (a);
\draw [bend right] (an1s) to (a);
\end{tikzpicture}
\caption{Illustration of the gadget that makes sure that every solution containing $a$ also contains $f_a$, $g_a$ and $a'$.
The vertex $a$ is a non-necessary, non-forbidden vertex from the \SSFN\ instance and may have other neighbors from this instance.
The vertex $a'$ additionally has the neighbors depicted in Figure~\ref{fig:necessary-reduction-gadget2}.}
\label{fig:necessary-reduction-gadget1}
\end{figure}
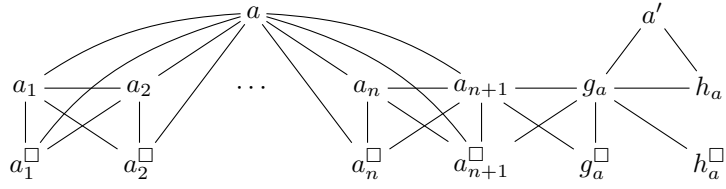

%% file: figures/necessary-reduction-gadget2.tex
\begin{figure}
\centering
\begin{tikzpicture}
[inner sep=2pt,bend angle=20]
\node (x) at (2,2) {$x$};

\node (x1) at (0,1) {$x_1$};
\node (x2) at (1,1) {$x_2$};
\node (xdots) at (2,1) {$\cdots$};
\node (xn) at (3,1) {$x_n$};
\node (xn1) at (4,1) {$x_{n+1}$};

\node (x1s) at (0,0) {$x_1^\square$};
\node (x2s) at (1,0) {$x_2^\square$};
\node (xns) at (3,0) {$x_n^\square$};
\node (xn1s) at (4,0) {$x_{n+1}^\square$};

\node (y) at (6,2) {$y$};

\node (y1) at (5,1) {$y_1$};
\node (ydots) at (6,1) {$\cdots$};
\node (yn1) at (7,1) {$y_{n+1}$};

\node (y1s) at (5,0) {$y_1^\square$};
\node (yn1s) at (7,0) {$y_{n+1}^\square$};

\node (ap) at (9,2) {$a'$};

\node (ap1) at (8,1) {$a'_1$};
\node (apdots) at (9,1) {$\cdots$};
\node (apn1) at (10,1) {$a'_{n+1}$};

\node (ap1s) at (8,0) {$a'^\square_1$};
\node (apn1s) at (10,0) {${a'^\square_{n+1}}$};

\node (bp) at (12,2) {$b'$};

\node (bp1) at (11,1) {$b'_1$};
\node (bpdots) at (12,1) {$\cdots$};
\node (bpn1) at (13,1) {$b'_{n+1}$};

\node (bp1s) at (11,0) {$b'^\square_1$};
\node (bpn1s) at (13,0) {$b'^\square_{n+1}$};

\draw (x2s) -- (x1) -- (x2) -- (x1s);
\draw (xn) -- (xn1) -- (y1);
\draw (xn) -- (xn1s) -- (y1);
\draw (xns) -- (xn1) -- (y1s);
\draw (ap1s) -- (yn1) -- (ap1) -- (yn1s);
\draw (bp1s) -- (apn1) -- (bp1) -- (apn1s);

\foreach \x in {x} {
  \draw [bend right] (\x) to (\x 1);
  \draw [bend right] (\x) to (\x 1s);
  \draw [bend left] (\x) to (\x n1);
  \draw [bend left] (\x) to (\x n1s);

  \draw (\x) -- (\x 2);
  \draw (\x) -- (\x 2s);
  \draw (\x) -- (\x n);
  \draw (\x) -- (\x ns);

	\foreach \y in {1,2,n,n1} {
		\draw (\x\y) -- (\x\y s);
	}
}

\foreach \x in {y,ap,bp} {
  \draw (\x) to (\x 1s);
  \draw (\x) to (\x n1s);

	\foreach \y in {1,n1} {
    \draw (\x) -- (\x\y) -- (\x\y s);
  }
}
\end{tikzpicture}
\caption{Illustration of the gadget that makes sure that every solution contains all necessary vertices if it contains some necessary vertex or if it contains $v'$ for some non-necessary vertex $v$.
Here we assume there are the four vertices $a,b,x,y$, among which $x$ and $y$ are necessary, and we use the ordering $x \preceq y \preceq a \preceq b$.}
\label{fig:necessary-reduction-gadget2}
\end{figure}
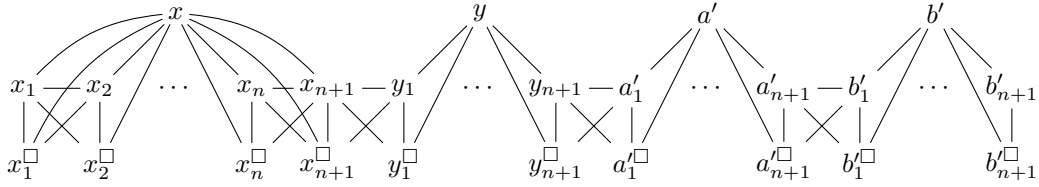